\documentclass[conference]{IEEEtran}
\usepackage{comment}
\usepackage{multicol}

\usepackage{blindtext}
\usepackage{amssymb}
\usepackage{subfigure}
\usepackage{amsmath}
\usepackage{amsfonts}
\usepackage{mathrsfs}
\usepackage{engord}
\usepackage[dvips]{graphicx}
\usepackage{bbm}
\usepackage{threeparttable}
\usepackage{booktabs}
\usepackage{epsfig}
\usepackage{color}
\usepackage{graphicx}
\usepackage{multirow}
\usepackage{cite}
\usepackage{epstopdf}
\usepackage{amsthm}

\usepackage{framed}
\usepackage{url}

\usepackage[ruled,norelsize]{algorithm2e}

\makeatletter
\newcommand{\removelatexerror}{\let\@latex@error\@gobble}
\makeatother
\begin{document}

\newtheorem{thm}{Theorem}
\newtheorem{prop}{Proposition}
\newtheorem{reproof}{Proof}
\newtheorem{lem}{Lemma}
\newtheorem{defn}{Definition}
\newtheorem{ex}{Example}
\newtheorem{cor}{Corollary}
\newtheorem{prn}{Principle}
\newtheorem{case}{Case}
%
\title{Latent Factor Analysis of Gaussian Distributions under  Graphical Constraints}

\author{\IEEEauthorblockN{Md Mahmudul Hasan,
Shuangqing Wei, Ali Moharrer}}
\maketitle
\footnotetext[1]{Md M Hasan,   S. Wei and A. Moharrer  are with the school of Electrical Engineering and Computer Science, Louisiana State University, Baton Rouge, LA 70803, USA (Email: mhasa15@lsu.edu, swei@lsu.edu, alimoharrer@gmail.com). }



\begin{abstract}
In this paper, we explore the algebraic structures of  solution spaces for Gaussian latent factor analysis when the population covariance matrix $\Sigma_x$ has an additional latent graphical constraint, namely, a latent star topology. In particular, we give sufficient and necessary conditions under which the solutions to constrained minimum trace factor analysis (CMTFA)  is still star.  We further show that the solution to CMTFA under the star constraint can only have two cases, i.e. the number of latent variable can be only one (star) or $n-1$ where $n$ is the dimension of the observable vector, and characterize the solution for both the cases. 
\end{abstract}
\section{INTRODUCTION}
Factor Analysis (FA) is a commonly used tool in multivariate statistics to represent the correlation structure of a set of observables in terms of significantly smaller number of variables called 
``latent factors". With the growing use of data mining, high dimensional data and analytics, factor analysis has already become a prolific area of research \cite{chen2017structured}\cite{bertsimas2017certifiably}. Classical Factor Analysis models seek to decompose the correlation matrix of an $n$-dimensional  random vector ${\bf X} \in {\mathcal R}^n$, $\Sigma_x  $, as the sum of a diagonal matrix $ D $ and a Gramian matrix $ \Sigma_{x}-D $. 

The literature that approached Factor Analysis can be classified in three major categories. Firstly, algebraic approaches \cite{albert1944matrices} and  \cite{drton2007algebraic}, where the principal aim was to give a characterization of the vanishing ideal of the set of symmetric $ n\times n $ matrices that decompose as the sum of a diagonal matrix and a low rank matrix, did not offer scalable algorithms for higher dimensional statistics. Secondly, Factor Analysis  via heuristic local optimization techniques, often based on the expectation maximization algorithm, were computationally tractable  but offered no provable performance guarantees. The third and final type of approach, based on convex optimization methods namely Minimum Trace Factor Analysis (MTFA)\cite{ledermann1940problem} and Minimum Rank Factor Analysis (MRFA)\cite{harman1976modern}, guarranted performance and were computationally tractable. As the name suggests MRFA seeks to minimize the rank of $ \Sigma_{x}-D $ and MTFA minimizes the trace of $ \Sigma_{x}-D $. Ideally rank minimization approaches would lead to the least number of latent factors but they are coptutationally much more challenging than trace minimization approaches. Trace as an objective function is almost as effective as the rank of a matrix and at the same time computationally tractable. Trace minimization is favored over rank minimization because of the fact that the trace of a matrix being a continuous function  offers more flexibility than the rank of matrix which is a discrete function. However,  MTFA solution could lead to negative values for the diagonal entries of the matrix $ D $. To solve this problem Constrained Minimum Trace Factor Analysis (CMTFA) was proposed \cite{bentler1980inequalities}, which imposes extra constraint of requiring $ D $ to be Gramian. Computational aspects of CMTFA and uniqueness of its solution was discussed in \cite{ten1981computational}.

Gaussian graphical models \cite{gomez2014sensitivity} \cite{larranaga2013review} \cite{xiu2018multiple} have enjoyed  wide variety  of applications in  economics \cite{dobra2010modeling},  biology \cite{ahmed2008time} \cite{durbin1998biological},  image recognition \cite{besag1986statistical} \cite{geman1984stochastic}, social networks \cite{vega2007complex} \cite{wasserman1994social}  and many other fields. Among the Gaussian graphical models, we are particularly interested in the Gaussian latent tree models \cite{shiers2016correlation} where the ovservables are the leaves of the tree and the unovserved variables are the interior nodes. In the simplest form a Gausian latent tree with just one node is a 'star'. Gaussian latent trees are highly favored because of their sparce structure \cite{mourad2013survey} and the availability of computationally efficient  algorithms to learn their underlying topologies \cite{choi2011learning} \cite{saitou1987neighbor}.

The scope of this paper is limited to finding a close form solution to CMTFA problem and recovering the underlying graphical structure. It is important to remark that, our work is not concerned about the algorithm side of CMTFA which is already in literature. Rather, our focas is to characterize and find insights about the solution space of CMTFA. The most closely related works to our work are \cite{moharrer2017algebraic} and \cite{saunderson2012diagonal}. Moharrer and Wei  in \cite{moharrer2017algebraic} established  relationship between the common information problem \cite{wyner1975common} and MTFA, and named the problem Constrained Minimum Determinant Factor Analysis (CMDFA).  In \cite{saunderson2012diagonal}, a sufficient condition was found on the subspace of $ \Sigma_{x} $ for MTFA solution of $ \Sigma_{x} $ to be a star when $\Sigma_x$ is equipped with a latent star graphical constraint (i.e. using a single latent variable could interpret the correlation entries in $\Sigma_x$).  One of our contributions is that we have proved that the sufficient condition found in \cite{saunderson2012diagonal} for MTFA recovers the star topology  when $\Sigma_x$ has a star constraint  is not only sufficient, but also necessary, for the CMTFA problem.  For clarification, the recovery of a star topology under a star constraint in a factor analysis problem means  the resulting decomposition of $\Sigma_x$  ends up with $\Sigma_x-D$  having rank one.  Moreover, we also fully characterized the solution to CMTFA under a star constraint  for situations where the recovery of the latent star fails. In particular, we proved that there are only two possible solutions to the CMTFA problem under a latent star constraint, one of which is the recovery of the star (i.e. the optimal number of latent variable is $k=1$), and the other with the optimal number of latent variables $k=n-1$. Sufficient and necessary conditions are found for both cases. 

It should be noted that our focus in this paper is on the analytical solutions to the CMTFA problem under a latent star constraint. The insights obtained in this study will play a critical role when seeking analytical results of the factor analysis problems when $\Sigma_x$ has more general latent tree structure, which is under investigation and will be presented in our future works.

The rest of the paper is organized as follows: section II gives the formulation of the problem and general outline to the solution.  Necessary and sufficient conditions for two possible CMTFA solutions of $ \Sigma_{x} $ are given respectively in section III and IV. The conclusion, appendices and references follow at the end. 
 \section{Problem Formulation and General Outline to the Solution}
 
 First of all we define the real column vector $ \vec{\alpha} $ as $\vec{\alpha} = [\alpha_1, \dots, \alpha_n]' \in {\mathcal R}^n$ where  $ 0< |\alpha_{j}|< 1 $, $ j= 1,2, \dots ,n $. We further remark that if for any  element of $ \vec{\alpha} $ the following condition holds, we-call it a \textit{non-dominant} element, otherwise its a \textit{dominant} element. 
\begin{equation}\label{27}
|\alpha_{i}|\leq \sum_{j\neq i} |\alpha_{j}| \quad \quad i=1,2,\dots,n
\end{equation}
It is easy to see that there can be only one dominant element in a vector and that has to be the element with the biggest absolute value among all. We call $ \vec{\alpha} $ \textit{dominant} if its biggest element in terms of absolute value is dominant, otherwize $ \vec{\alpha} $ is \textit{non-dominant}. For the remainder of this paper, without the loss of generality, we assume that,
\begin{align}\label{order}
|\alpha_{1}|\geq |\alpha_{2}| \geq \dots \geq |\alpha_{n}|
\end{align}
Hence,  all dominance is defined with respect to $ |\alpha_{1} |$. Which implies that vector $ \vec{\alpha} $ is non dominant if the following holds,
\begin{equation}\label{ckf}
|\alpha_{1}|\leq \sum_{j=2}^{n} |\alpha_{j}|
\end{equation}
Otherwise, vector $ \vec{\alpha} $ is dominant.

 Let us consider a star structured population covariance matrix $ \Sigma_{x} $ having all the diagonal comptonents $ 1 $  as given by equation \eqref{10}.
 \begin{equation}\label{10}
\Sigma_{x}=\begin{bmatrix}
1&\alpha_{1}\alpha_{2}&\dots&\alpha_{1}\alpha_{n}\\
\alpha_{2}\alpha_{1}&1&\dots&\alpha_{2}\alpha_{n}\\
\vdots&\vdots&\ddots&\vdots\\
\alpha_{n}\alpha_{1}&\alpha_{n}\alpha_{2}&\dots&1\\
\end{bmatrix}
\end{equation}

The above population matrix could be generated by the following latent structure. 

\begin{align}\label{3}
& \begin{bmatrix}
X_{1}\\
\vdots\\
X_{n}
\end{bmatrix}
=\begin{bmatrix}
\alpha_{1}\\
\vdots\\
\alpha_{n}
\end{bmatrix}
Y
+
\begin{bmatrix}
Z_{1}\\
\vdots\\
Z_{n}
\end{bmatrix}\\
&\Rightarrow \mathbf{X}= \vec{\alpha}Y+\mathbf{Z}
\end{align}
where
\begin{itemize}
\item $ \{X_{1}, ..., X_{n}\}$  are conditionally independent Gaussian random variables given $ Y $, forming the jointly Gaussian random vector $ \mathbf{X}\sim \mathcal{N}(0,\Sigma_{x}) $ where $ Y\sim \mathcal{N}(0,1) $.
\item $ \{Z_{1}, ..., Z_{n} \}$ are independent Gausian random varables with $Z_{j}\sim \mathcal{N}(0,1-\alpha_{j}^{2})\quad 1\leq j \leq n $ forming the Gaussian random vector $ \mathbf{Z} $.
\item $ \Sigma_{z} $ is the covariance  matrix of vector $\mathbf{Z}$ of  independent Gausian random varables $\{Z_{1}, ..., Z_{n}\}$.
\end{itemize}
CMTFA seeks to decompose $ \Sigma_{x} $ as,
\begin{align}\label{cmtfa}
\Sigma_{x}=(\Sigma_{x}-D)+D
\end{align}
such that the trace of $ (\Sigma_{x}-D) $ is maximized or equivalently the trace of $ D $ is minimized under the constraint that both $ (\Sigma_{x}-D) $ and $ D $ are Gramian matrices. Let $ D^{*} $ be the CMTFA solution of $ \Sigma_{x} $ and $ d^{*} $ be the $ n $ dimentional column  vector with each entry being the corresponding diagonal entry of the matrix $ D^{*} $. The following necessary and sufficient condition for $ d^{*} $ to be the CMTFA solution of $ \Sigma_{x} $ was set in \cite{della1982minimum}, 

The point $ d^{*} $ is a solution of the CMTFA  problem if and only if $ d^{*}\geq 0 $, $ \lambda(d^{*})=0 $ which is the minimum eigenvalue of $ (\Sigma_{x}-D^{*}) $,  and there exist $ \vec{t}_{i}\in N(\Sigma_{x}-D^{*}), \quad i=1,....,r $ where $ N(\Sigma_{x}-D^{*})$ is the null space of the matrix $(\Sigma_{x}-D^{*})$ such that the following holds,
\begin{align}\label{13}
\mathbf{1}=\sum_{i=1}^{r}\vec{t}_{i}^{2} - \sum_{j\in I(d^{*})}\mu_{j}\vec{\xi_{j}}
\end{align}
where $ r\leq n $, $ \mathbf{1} $ is $ n $ dimensional column vector with all the components equal to $ 1 $,  $ \{\vec{t}_{i}\in N(\Sigma_{x}-D^{*}), \quad i=1,....,r \}$ are $ n $ dimensional column vectors forming the rank $ (n-k) $ matrix $ T$ , $ \vec{t}_{i}^{2} $ is the Hadamard product of vector $ \vec{t}_{i} $ with itself,   $ I(d^{*})=\{i:d_{i}^{*}=0, i\leq n\} $, $ \{\mu_{j}, \quad j\in I(d^{*}) \}$  are non-negative numbers and $\{ \vec{\xi}_{j}, j\in I(d^{*})\}$ are column vectors in $ {\mathcal R}^{n} $ with all the components equal to $ 0 $ except for the $ j $th component which is equal to $ 1 $.

The problem that we are looking at can be stated as follows: we are trying to find a close form analytical solution for CMTFA problem and gain insights about the underlyting graphical structure. 
To be more specific our primary focus is to see if the underlying structure of  CMTFA solution to $ \Sigma_{x} $ with a star constraint is still a star or mathematically speaking to see  if $(\Sigma_{x}-D^{*}) $ is a rank one matrix given that $ D^{*} $ is the solution to \eqref{cmtfa}. 

Now we give a brief outline of our findings. 
We show that the CMTFA solution to $ \Sigma_{x} $  recovers the model given by \eqref{3}  if and only if  vector $ \vec{\alpha} $ is non-dominant. Equivalently speaking for such $ \vec{\alpha} $ CMTFA solution is a rank $ 1 $ matrix given by \eqref{28}.
\begin{align}\label{28}
\Sigma_{t,ND}=\begin{bmatrix}
\alpha_{1}^{2}&\alpha_{1}\alpha_{2}&\dots&\alpha_{1}\alpha_{n}\\
\alpha_{2}\alpha_{1}&\alpha_{2}^{2}&\dots&\alpha_{2}\alpha_{n}\\
\vdots&\vdots&\ddots&\vdots\\
\alpha_{n}\alpha_{1}&\alpha_{n}\alpha_{2}&\dots&\alpha_{n}^{2}\\
\end{bmatrix}
\end{align}

 We also show that if CMTFA solution of $ \Sigma_{x} $ is not a star i.e., if it is not a rank $ 1 $ solution, it can only be a rank $ n-1 $ solution characterized by \eqref{sigmat}.
 
 \begin{equation}\label{sigmat}
\Sigma_{t,DM}=\begin{bmatrix}
(\Sigma_{t,DM})_{11}&\alpha_{1}\alpha_{2}&\dots&\alpha_{1}\alpha_{n}\\
\alpha_{2}\alpha_{1}&(\Sigma_{t,DM})_{22}&\dots&\alpha_{2}\alpha_{n}\\
\vdots&\vdots&\ddots&\vdots\\
\alpha_{n}\alpha_{1}&\alpha_{n}\alpha_{2}&\dots&(\Sigma_{t,DM})_{nn}\\
\end{bmatrix}
\end{equation}
where
\begin{align}
&(\Sigma_{t,DM})_{11}=|\alpha_{1}|\left(\sum_{i\neq 1}|\alpha_{i}|\right)\notag\\
&(\Sigma_{t,DM})_{ii}=|\alpha_{i}|\left(|\alpha_{1}|-\sum_{j\neq i,1}|\alpha_{j}|\right), \quad i=2,\dots, n\notag
\end{align} 

We will elaborate on the above two solutions and their proofs in the follwoing two sections. 
\section{Dominant Case}

In this section we analyse the conditions under which the CMTFA solution of $ \Sigma_{x} $ given by \eqref{10} is not a star. Theorem $ 1 $ states the main outcome of this section. 
\begin{thm}\label{th1}
$ \Sigma_{t,DM} $ given by equation \eqref{sigmat} is the CMTFA solution of $ \Sigma_{x} $ if and only if $ \vec{\alpha} $ is dominant. 
\end{thm}
Before we embark on the proof of Theorem $ 1 $, understanding the
following two Lemmas are of significant importance.

\begin{lem}\label{l1}
$ \Sigma_{t,DM} $ is a rank $ n-1 $ matrix. 
\end{lem}
\begin{proof}[\textbf{Proof of Lemma \ref{l1}:}]
Let $ \gamma_{i}\in \{-1 , 1\} $ be the sign of $ \alpha_{i} $, i.e. $ \alpha_{i}=\gamma_{i}|\alpha_{i}| $.

For the $ 1 $st column of $ \Sigma_{t,DM} $,
\begin{align}
\sum_{g=2}^{n} \gamma_{1} \gamma_{g}(\Sigma_{t,DM})_{g1}&=\sum_{g=2}^{n} \gamma_{1} \gamma_{g} \gamma_{1} \gamma_{g}|\alpha_{g}||\alpha_{1}|\notag\\
&=\sum_{g=2}^{n} |\alpha_{g}||\alpha_{1}|\notag\\
&=|\alpha_{1}|\left(\sum_{g=2}^{n}|\alpha_{g}|\right)=(\Sigma_{t,DM})_{11}\notag
\end{align}

For the $ h $th $ (h\neq 1) $ column  of $ \Sigma_{t,DM} $,
\begin{align}
&\sum_{g=2}^{n} \gamma_{1} \gamma_{g} (\Sigma_{t,DM})_{gh}\notag\\
&=\gamma_{1} \gamma_{h}|\alpha_{1}||\alpha_{h}|-\sum_{m\neq h,1} \gamma_{1} \gamma_{h}|\alpha_{h}||\alpha_{m}|\notag\\
&+\sum_{m\neq h,1} \gamma_{1} \gamma_{m} \gamma_{m} \gamma_{h}|\alpha_{h}||\alpha_{m}|\notag\\
&=\gamma_{1} \gamma_{h}|\alpha_{1}||\alpha_{h}|-\sum_{m\neq h,1} \gamma_{1} \gamma_{h}|\alpha_{h}||\alpha_{m}|\notag\\
&+\sum_{m\neq h,1} \gamma_{1} \gamma_{h}|\alpha_{h}||\alpha_{m}|\notag\\
&=\gamma_{1} \gamma_{h}|\alpha_{1}||\alpha_{h}|\notag\\
&=(\Sigma_{t,DM})_{1h}\notag
\end{align}

Combining the above two results,

\begin{align}
&(\Sigma_{t,DM})_{1}=\sum_{g=2}^{n} \gamma_{1} \gamma_{g}(\Sigma_{t,DM})_{g}\notag\\
& \Rightarrow(\Sigma_{t,DM})_{1}-\sum_{g=2}^{n} \gamma_{1} \gamma_{g}(\Sigma_{t,DM})_{g}=0\notag\\
&\Rightarrow(\Sigma_{t,DM})_{1}-\sum_{g=2}^{n} (-1)^{S_{g}}(\Sigma_{t,DM})_{g}=0\notag
\end{align}
where 
\begin{equation*}
    S_{g}=
    \begin{cases}
      1, & \gamma_{1} \gamma_{g}=-1\\
    2, & \gamma_{1} \gamma_{g}=1\\
        \end{cases}
  \end{equation*}
 \end{proof}

\begin{lem}\label{l2}
There exists a column vector $\mathbf{\Phi}=[\Phi_{1}, \Phi_{2}, .... , \Phi_{n}]' $ such that $ \Sigma_{t,DM}\mathbf{\Phi}=0$, where $ \Phi_{i}\in \{-1, 1\}, 1\leq i \leq n $.
\end{lem}
\begin{proof}[\textbf{Proof of Lemma \ref{l2}:}]
It is obvious to see that the following selection of the elements of vector  $ \mathbf{\Phi} $ makes it orthogonal to $ (\Sigma_{t,DM})_{1} $, i.e. $ (\Sigma_{t,DM})_{1}\mathbf{\Phi}=0 $. Where $ (\Sigma_{t,DM})_{1} $ is the $ 1 $st row of $ \Sigma_{t,DM} $.

\begin{equation*}
    \Phi_{i}=
    \begin{cases}
      -1, & \alpha_{1} \alpha_{i}>0, i\neq 1\\
    1, & otherwise\\
        \end{cases}
  \end{equation*}
Now it will be sufficient to prove that any vector $ \mathbf{\Phi} $ orthogonal to $ (\Sigma_{t,DM})_{1} $ is also orthogonal to all the other rows of $ \Sigma_{t,DM} $, i.e. $ (\Sigma_{t,DM})_{i}\mathbf{\Phi}=0, 2\leq i\leq n $. 

Let $ \gamma_{i}\in \{-1 , 1\} $ be the sign of $ \alpha_{i} $, i.e. $ \alpha_{i}=\gamma_{i}|\alpha_{i}| $

Now for any row $ g $, $ g\neq 1 $,
\begin{align}
(\Sigma_{t,DM})_{g}\mathbf{\Phi}&=\Phi_{g}(\Sigma_{t,DM})_{gg}+\sum_{g\neq h} \Phi_{h}(\Sigma_{t,DM})_{gh}\notag\\
&=\Phi_{g} |\alpha_{g}|\left(|\alpha_{1}|-\sum_{i\neq g,1}|\alpha_{i}|\right)+\sum_{g\neq h} \Phi_{h}\alpha_{g}\alpha_{h}\notag\\
&=\Phi_{g} |\alpha_{g}||\alpha_{1}|+\Phi_{1} \alpha_{g}\alpha_{1}-\sum_{i\neq g, i\neq 1} \Phi_{g}|\alpha_{g}||\alpha_{i}|\notag\\
&+\sum_{h\neq g, h\neq 1} \Phi_{h}\alpha_{g}\alpha_{h}\notag\\
&=(\Phi_{g} +\Phi_{1}\gamma_{g}\gamma_{1})| \alpha_{g}||\alpha_{1}|\notag\\
&+\sum_{h\neq g, h\neq 1} (\gamma_{g}\gamma_{h}\Phi_{h}-\Phi_{g})|\alpha_{g}||\alpha_{h}|\label{3aa}
\end{align}
If $ \Phi_{g}=\Phi_{h} \Rightarrow \gamma_{1}\gamma_{g}=\gamma_{1}\gamma_{h} \Rightarrow \gamma_{g}= \gamma_{h}\Rightarrow \gamma_{g}\gamma_{h}\Phi_{h}-\Phi_{g}=0 $. \\

Else if $ \Phi_{g}\neq \Phi_{h} \Rightarrow \gamma_{1}\gamma_{g}\neq \gamma_{1}\gamma_{h} \Rightarrow \gamma_{g}\neq \gamma_{h}\Rightarrow \gamma_{g}\gamma_{h}\Phi_{h}-\Phi_{g}=0 $.\\

Similarly, \
If $ \Phi_{g}=\Phi_{1}\Rightarrow \alpha_{1}\alpha_{g}<0 \Rightarrow \gamma_{1}\neq \gamma_{g}\Rightarrow \Phi_{g} +\Phi_{1}\gamma_{g}\gamma_{1} =0$\\

Else if $ \Phi_{g}\neq \Phi_{1}\Rightarrow \alpha_{1}\alpha_{g}>0 \Rightarrow \gamma_{1}= \gamma_{g}\Rightarrow \Phi_{g} +\Phi_{1}\gamma_{g}\gamma_{1} =0$\\

Plugging these results in  equation \eqref{3aa}, we get
\begin{align}
(\Sigma_{t,DM})_{g}\mathbf{\Phi}=0\notag
\end{align}
\end{proof}
Having proved the two Lemmas, we are now well equipped to prove  Theorem $ 1 $. 

\begin{proof}[\textbf{Proof of Theorem \ref{th1}}]
To prove the Theorem we refer to necessary and sufficient condition set in \eqref{13}.
Rank of $ \Sigma_{t,DM} $ is $ n-1 $, so its minimum eigenvalue is $ 0 $. Since each $0< |\alpha_{i}|<1 $, $ 0<(\Sigma_{t,DM})_{ii})<1, i=1,\dots, n $. Hence all the diagonal entries $ d_{i} $ of $ D $ are positive. As a result, the set $I(d^{*})$ is empty and the second term in the right hand side of \eqref{13} vanishes.

The dimension of the null space of $ \Sigma_{t,DM} $ is $ 1 $. It will suffice for us to prove the existence of a  column vector $ \mathbf{\Phi}_{n\times 1} $, $\Phi_{i} \in \{1,-1\}, 1\leq i \leq n $ such that $ \Sigma_{t,DM}\mathbf{\Phi}=0 $. Lemma $ 2 $ gives that proof. 
\end{proof}
\section{Non-Dominant Case}
This section is dedicated to the analytical details of the conditions under which the CMTFA solution of star structured $ \Sigma_{x} $ is also star. Theorem $ 2 $ states the main outcome of the section. 
\begin{thm}\label{th2}
$ \Sigma_{t,ND} $ is the CMTFA solution of $ \Sigma_{x}$ if and only if $ \vec{\alpha} $ is non-dominant.
\end{thm}
The theorem states that  the CMTFA solution to a star connected network is a star itself if and only if there is no dominant element in the vector $ \vec{\alpha} $. The knowledge of the following  Lemma is of significant importance, before we embark on the proof of Theorem $ 2 $.

\begin{lem}\label{l3}
There exists rank $ n-1 $ matrix $ T_{n \times n} $ such that the column vectors of $ T$ are in the null space of $ \Sigma_{t,ND} $ and the $ L_{2} $-norm of each row of $ T $ is $ 1 $. 
\end{lem}
\begin{proof}[\textbf{Proof of Lemma \ref{l3}}]
Its trivial to find  the following  basis vectors for the null space of  $ \Sigma_{t,ND} $,

\begin{align}
v_{1}=\begin{bmatrix}
-\frac{\alpha_{2}}{\alpha_{1}}\\
1\\
0\\
\vdots\\
0
\end{bmatrix}, \qquad v_{2}=\begin{bmatrix}
-\frac{\alpha_{3}}{\alpha_{1}}\\
0\\
1\\
\vdots\\
0
\end{bmatrix},\dots, \qquad v_{n-1}=\begin{bmatrix}
-\frac{\alpha_{n}}{\alpha_{1}}\\
0\\
0\\
\vdots\\
1
\end{bmatrix}
\end{align}

We define matrix $ V $ so that its columns span the null space of $ \Sigma_{t,ND} $,
\begin{align}\label{v}
&V= \begin{bmatrix}
-\frac{\alpha_{2}}{\alpha_{1}}&-\frac{\alpha_{3}}{\alpha_{1}}&\dots&-\frac{\alpha_{n}}{\alpha_{1}}&-\left(c_{2}\frac{\alpha_{2}}{\alpha_{1}}+\dots+c_{n}\frac{\alpha_{n}}{\alpha_{1}}\right)\\
1&0&\dots&0&c_{2}\\
0&1&\dots&0&c_{3}\\
\vdots&\vdots&\ddots&\vdots&\vdots\\
0&0&\dots&1&c_{n}
\end{bmatrix}
\end{align} 
To prove the Lemma , it will suffice for us to show the existance of  $ \{c_{j}\} \quad 1\leq j \leq n $ and a diagonal matrix $ B_{n \times n} $ such that the following holds. 
\begin{align}\label{9}
T_{n \times n}= V_{n \times n}B_{n \times n}
\end{align}
where, $ L_{2} $-norm of each row of $ T $ is $ 1 $. Using \eqref{9}, 
\begin{align}
TT'= VBB'V'
\end{align}

We define the symmetric matrix $ \beta=BB' $, and we require the diagonal  matrix $ \beta $ to have only non-negative entries. 

Since we want each diagonal element of $ TT' $ to be $ 1 $, we have the following $ n $ equations,

\begin{align}\label{e1}
&\frac{\alpha_{2}^{2}}{\alpha_{1}^{2}}\beta_{11}+\frac{\alpha_{3}^{2}}{\alpha_{1}^{2}}\beta_{22}+\dots+\frac{\alpha_{n}^{2}}{\alpha_{1}^{2}}\beta_{n-1,n-1}+\notag\\
&\left(c_{2}\frac{\alpha_{2}}{\alpha_{1}}+c_{3}\frac{\alpha_{3}}{\alpha_{1}}+\dots+c_{n}\frac{\alpha_{n}}{\alpha_{1}}\right)^{2}\beta_{nn}=1
\end{align}

\begin{align}\label{e2}
\beta_{ii}+c_{i+1}^{2}\beta_{nn}=1, \quad i=1,\dots,n-1
\end{align}

Solving \eqref{e1}, we get, 
\begin{align}
&\beta_{nn}=
\frac{\alpha_{1}^{2}-\alpha_{2}^{2}-\alpha_{3}^{2}-\dots-\alpha_{n}^{2}}{\sum_{i\neq j, i\neq 1, j\neq 1}c_{i}c_{j}\alpha_{i}\alpha_{j}}
\end{align}

Since the diagonal entries of $ \beta $ can only be non-negative, we have the following three cases.
\begin{align*}
\alpha_{1}^{2}-\alpha_{2}^{2}-\alpha_{3}^{2}-\dots-\alpha_{n}^{2} =0\\
\alpha_{1}^{2}-\alpha_{2}^{2}-\alpha_{3}^{2}-\dots-\alpha_{n}^{2} >0\\
\alpha_{1}^{2}-\alpha_{2}^{2}-\alpha_{3}^{2}-\dots-\alpha_{n}^{2} <0
\end{align*} 
It will suffice for us to prove that for all of the above cases there exist $\{c_j, 2 \leq j \leq n\}$ that make $ \beta_{ii}\geq 0, i=1,\dots,n $. 

\textbf{Case $ 1 $:} It is straightforward. If $\alpha_{1}^{2}-\alpha_{2}^{2}-\alpha_{3}^{2}-\dots-\alpha_{n}^{2} =0 $\\
Then using \eqref{e1} and \eqref{e2} we get, $ \beta_{nn}=0 $ and $ \beta_{11}=\beta_{22}=\dots= \beta_{n-1,n-1}=1 $. \\
In line with the assumption we made, normalizing each element in \eqref{order} by $ |\alpha_{1}| $ gives us the  following, 
\begin{align}
1\geq |\widetilde{\alpha}_{2}| \geq |\widetilde{\alpha}_{3}| \geq  \dots \geq |\widetilde{\alpha}_{n}|
\end{align}
where $ \widetilde{\alpha}_{j}=\frac{\alpha_{j}}{\alpha_{1}}, 1\leq j\leq n $. 
We define,
\begin{align}\label{smin}
S_{min}=\min_{A}\left\vert \sum_{j\in A}|\widetilde{\alpha}_{j}| -\sum_{j\in A^{c}}|\widetilde{\alpha}_{j}|\right\vert
\end{align}
where $ A\subset \{2,3, \dots ,n \} $ and $ A^{c}=\{2,3, \dots ,n \}-A $

Let $ A^{*} $ be the event that gives us $ S_{min} $ over all the possible events of the set $ A $ in \eqref{smin}. Assuming that $ A^{*} $ has $ l $ elements, let the set $ A^{*} $  be $ A^{*}=\{a_{1}, a_{2},..., a_{l}\} $. We define the set  $ F $ as,
\begin{align}
F=\{F_{a_{1}},\dots,F_{a_{l}}, \quad F_{a_{i}}=|\widetilde{\alpha}_{a_{i}}|, a_{i}\in A^{*}\}\notag
\end{align}
Now under this ordered and normalized settings, we have\\
\textbf{Case} $ 2 $: $1-\widetilde{\alpha}_{2}^{2}-\widetilde{\alpha}_{3}^{2}-\dots-\widetilde{\alpha}_{n}^{2} > 0 $.\\
We can select $ c_{2},c_{3},\dots,c_{n}$ in a way such that $ c_{i}\widetilde{\alpha}_{i}=|\widetilde{\alpha}_{i}| $ to make $ \beta_{nn} > 0$. Equation \eqref{e2} dictates that to ensure the other diagonal entries of $ \beta $ are non-negative, the following must hold, 
\begin{align}
&\frac{1-\widetilde{\alpha}_{2}^{2}-\widetilde{\alpha}_{3}^{2}-\dots-\widetilde{\alpha}_{n}^{2}}{\sum_{i\neq j, i\neq 1, j\neq 1}c_{i}c_{j}\widetilde{\alpha}_{i}\widetilde{\alpha}_{j}}\leq 1\notag\\
\Longleftrightarrow &1\leq (|\widetilde{\alpha}_{2}|+|\widetilde{\alpha}_{3}|+\dots+|\widetilde{\alpha}_{n}|)^{2}\notag\\
\Longleftrightarrow &1\leq |\widetilde{\alpha}_{2}|+|\widetilde{\alpha}_{3}|+\dots+|\widetilde{\alpha}_{n}|\label{d1}
\end{align}
which means such $ \beta_{nn} $  exists if and only if $ \widetilde{\alpha}_{1} $ is non-dominat. Because of the ordered representation, that essentially means $ \vec{\alpha} $ has to be non-dominant.

\textbf{Case $ 3 $:} $1-\widetilde{\alpha}^{2}-\dots-\widetilde{\alpha}_{n}^{2} < 0 $\\
Using the Lemma $ 6 $ given  in Appendix B of this paper, if we select $ c_{i}\in\{1,-1\} $ such that $\sum_{j=2}^{n}c_{j}\widetilde{\alpha}_{j}=S_{min}$ then, $\sum_{i\neq j, i\neq 1, j\neq 1}c_{i}c_{j}\widetilde{\alpha}_{i}\widetilde{\alpha}_{j} <0$. And for such selection of  $ c_{i} $ we have, 
\begin{align}
&\widetilde{\alpha}_{2}^{2}+\widetilde{\alpha}_{3}^{2}+\dots+\widetilde{\alpha}_{n}^{2}+\sum_{i\neq j, i\neq 1, j\neq 1}c_{i}c_{j}\widetilde{\alpha}_{i}\widetilde{\alpha}_{j}\notag\\
&=\left(\sum_{i=2}^{n}c_{i}\widetilde{\alpha}_{i}\right)^{2}=S_{min}^{2}\leq 1
\end{align}
The last inequality is due to Lemma $ 4 $ given in Appendix A of this paper, that shows $ S_{min}\leq F_{a_{i}}, a_{i}\in A^{*} $. 
So, we have, 
\begin{align}
&\widetilde{\alpha}_{2}^{2}+\widetilde{\alpha}_{3}^{2}+\dots+\widetilde{\alpha}_{n}^{2}+\sum_{i\neq j, i\neq 1, j\neq 1}c_{i}c_{j}\widetilde{\alpha}_{i}\widetilde{\alpha}_{j}\leq 1\notag\\
\Longleftrightarrow & 1-\widetilde{\alpha}_{2}^{2}-\widetilde{\alpha}_{3}^{2}-\dots-\widetilde{\alpha}_{n}^{2}\geq \sum_{i\neq j, i\neq 1, j\neq 1}c_{i}c_{j}\widetilde{\alpha}_{i}\widetilde{\alpha}_{j}\notag
\end{align}

Both the terms $ 1-\widetilde{\alpha}_{2}^{2}-\widetilde{\alpha}_{3}^{2}-\dots-\widetilde{\alpha}_{n}^{2} $ and $ \sum_{i\neq j, i\neq 1, j\neq 1}c_{i}c_{j}\widetilde{\alpha}_{i}\widetilde{\alpha}_{j} $ are negative. 
Hence,
\begin{align}
\beta_{nn}=\frac{1-\widetilde{\alpha}_{2}^{2}-\widetilde{\alpha}_{3}^{2}-\dots-\widetilde{\alpha}_{n}^{2} }{\sum_{i\neq j, i\neq 1, j\neq 1}c_{i}c_{j}\widetilde{\alpha}_{i}\widetilde{\alpha}_{j} }\leq 1
\end{align}
\end{proof}
Having proved Lemma $ 3 $, we can now proceed to prove  Theorem $ 2 $. 

\begin{proof}[\textbf{Proof of Theorem \ref{th2}:}]
We still use the same necessary and sufficient condition set in \eqref{13}. $ \Sigma_{t,ND} $ is rank $ 1 $, so its minimum eigenvalue is $ 0 $.  Since each $0< |\alpha_{i}|<1 $, $ 1-\alpha_{i}^{2}>0, 1\leq i \leq n $. As a result the set $ I(d^{*})$ is empty. So, the second term on the right side of \eqref{13} vanishes. \\
The dimention of the null space of $ \Sigma_{t,ND} $ is $ n-1 $.  Lemma~\ref{l3} proves that   there exists rank $ n-1 $ matrix $ T_{n \times n} $ such that the column vectors of $ T$ are in the null space of $ \Sigma_{t,ND} $ and the $ L_{2} $-norm of each row of $ T $ is $ 1 $. That essentially completes the proof. 
\end{proof}
It is worthwhile to remark that in \cite{saunderson2012diagonal}  the condition of non dominance given by equation \eqref{ckf} was found as a sufficient condition for MTFA solution to be recoverable. We have proved through Theorem $ 2 $ that the condition given by \eqref{ckf} is both sufficient and necessary for CMTFA solution to recover the star structure. 

\section{Conclusion}
In this paper we characterized the solution space of  CMTFA. We showed that the  CMTFA solution of a star structured population matrix can have either a rank $ 1 $ or a rank $ n-1 $  solution and nothing in between. We found necessary and sufficient conditions for both of the solutions.

\appendices
\section{}

We have two appendices namely Appendix A and Appendix B. Appendix A has Lemma $ 4 $ and Appendix B has Lemmas $ 5 $ and $ 6 $. Before we state Lemma $ 4 $, we lay some ground work in terms of defining some parameters and notations which will remain consistent in Appendix B as well. 

Let $ e_{1}, e_{2}, \dots, e_{n}$ be a set of $ n $ positive numbers. 

We define,
\begin{align}\label{dd}
S_{min}=\min_{A}\left\vert \sum_{i\in A}e_{i} -\sum_{j\in A^{c}}e_{j}\right\vert
\end{align}
where $ A\subset \{1,2,3, \dots ,n \} $ and $ A^{c}=\{1,2,3, \dots ,n \}-A $

Let $ A^{*} $ be the event that gives us $ S_{min} $ over all the possible events of the set $ A $ in \eqref{dd}. Assuming the set $ A^{*} $  has $ l $ elements, let the sets $ A^{*} $ and $ (A^{*})^{c} $ be $ A^{*}=\{a_{1}, a_{2},..., a_{l}\} $ and $ (A^{*})^{c}=\{a_{1}^{c}, a_{2}^{c},..., a_{n-l}^{c}\} $. Let $ F $ and $ G $ be following two sets,
\begin{align}
F=\{F_{a_{1}},\dots,F_{a_{l}}, \quad F_{a_{i}}=e_{a_{i}}, a_{i}\in A^{*}\}\notag\\
G=\{G_{a_{1}^{c}},\dots,G_{a_{n-l}^{c}}, \quad G_{a_{i}^{c}}=e_{a_{i}^{c}}, a_{i}^{c}\in (A^{*})^{c}\}\notag
\end{align}
We define,
\begin{align}
&M+S_{min}=\sum_{a_{i}\in A^{*}}F_{a_{i}}, \quad M=\sum_{a_{i}^{c}\in (A^{*})^{c}}G_{a_{i}^{c}}\\
&F_{avg}=\frac{1}{l}(M+S_{min}), \quad G_{avg}=\frac{1}{n-l}M\\
&F_{min}=\min_{a_{i}\in A^{*}}F_{a_{i}}
\end{align}
\begin{lem}\label{l5}
$ S_{min}\leq F_{min} $
\end{lem}
\begin{proof}[\textbf{Proof of Lemma \ref{l5}:}]
Let us assume  $ F_{min}< S_{min} $ and has the value $ F_{min}=S_{min}-\epsilon$ where $0<\epsilon<S_{min} $. 

Now, If we deduct $ F_{min}$ from set $ F $ and add it to the set $ G $, then we will have,
\begin{align}
\left\vert (M+S_{min}-F_{min})-(M +F_{min}) \right\vert&=\left\vert S_{min}-2F_{min}\right\vert\notag\\
&=\left\vert S_{min}-2\epsilon\right\vert\notag\\
&<S_{min}\notag
\end{align}
which is not possible. So, $ F_{min}\geq S_{min}$.
\end{proof}
\section{•}
Appendix B has Lemmas $ 5 $ and $ 6 $. The proof of Lemma $ 6 $ depends on Lemma $ 5 $. The parameters and notations we laid in Appendix A hold their meaning in Appendix B as well. 
\begin{lem}\label{l6}
For the set of positive numbers $ e_{1},e_{2},\dots,e_{n} $,
\begin{align}
\sum_{i\neq j}e_{i}e_{j}\leq n(n-1)e_{avg}^{2}
\end{align}
where, $ e_{avg}=\frac{1}{n}\sum_{i=1}^{n}e_{i} $.
\end{lem}

\begin{proof}[\textbf{Proof of Lemma \ref{l6}:}]
Without the loss of generality, we can write the set of numbers in terms of their average in the following way: $e_{avg}+k_{1}, e_{avg}+k_{2},\dots, e_{avg}+k_{p},e_{avg}-j_{1}, e_{avg}-j_{2},\dots, e_{avg}-j_{q}$, where, $ p+q=n, k_{i}\geq 0, j_{i}\geq 0 $.

It is straightforward to see, $ \sum_{i=1}^{p}k_{i}=\sum_{i=1}^{q}j_{i} $.
We define, 
\begin{align}
\psi_{1}=&(e_{avg}+k_{1})\left[(e_{avg}+k_{2})+\dots+ (e_{avg}+k_{p})+\notag \right. \\ & \left.(e_{avg}-j_{1})+( e_{avg}-j_{2})+\dots+( e_{avg}-j_{q})\right]\notag\\
=&(e_{avg}+k_{1})\left[(p+q-1)e_{avg}+(k_{2}+\dots+k_{p})-\notag \right. \\ & \left.(j_{1}+\dots+j_{q})\right]\\
\psi_{2}=&(e_{avg}+k_{2})\left[(p+q-2)e_{avg}+(k_{3}+\dots+k_{p})\notag \right. \\ & \left.-(j_{1}+\dots+j_{q})\right]\\
&\vdots\notag\\
\psi_{p-1}=&(e_{avg}+k_{p-1})\left[(q+1)e_{avg}+k_{p} -(j_{1}+\dots+j_{q})\right]\\
\psi_{p}=&(e_{avg}+k_{p})\left[qe_{avg}-(j_{1}+\dots+j_{q})\right]\\
\psi_{p+1}=&(e_{avg}-j_{1})\left[(q-1)e_{avg}-(j_{2}+\dots+j_{q})\right]\\
&\vdots\notag\\
\psi_{p+q-1}=&(e_{avg}-j_{q-1})\left[e_{avg}-j_{q}\right]
\end{align}
Using the above equations,

\begin{align}
&\sum_{i\neq j}e_{i}e_{j}=2[\psi_{1}+\dots+\psi_{p+q-1}]\notag\\
&=2\left[(1+\dots+(p+q-1))e_{avg}^{2}+e_{avg}(p+q-1)\left(\sum_{i=1}^{p}k_{i}\right)\notag \right. \\ & \left.-e_{avg}(p+q-1)\left(\sum_{i=1}^{q}j_{i}\right) +\sum_{g=1}^{p-1}k_{g}\sum_{h=g+1}^{p}k_{h}\notag \right. \\ & \left.+\sum_{g=1}^{q-1}j_{g}\sum_{h=g+1}^{q}j_{h}-\left(\sum_{i=1}^{p}k_{i}\right)\left(\sum_{i=1}^{q}j_{i}\right)\right]
\end{align}
If $ \sum_{g=1}^{p-1}k_{g}\sum_{h=g+1}^{p}k_{h}\geq\sum_{g=1}^{q-1}j_{g}\sum_{h=g+1}^{q}j_{h} $
\begin{align}\label{33}
\sum_{i\neq j}e_{i}e_{j}&\leq 2\left[\frac{n(n-1)}{2}e_{avg}^{2} +2\sum_{g=1}^{p-1}k_{g}\sum_{h=g+1}^{p}k_{m}-\left(\sum_{i=1}^{p}k_{i}\right)^{2}\right]\notag\\
&=2\left[\frac{n(n-1)}{2}e_{avg}^{2} -\sum_{i=1}^{p}k_{i}^{2}\right]\notag\\
\end{align}

Else if, $ \sum_{g=1}^{p-1}k_{g}\sum_{h=g+1}^{p}k_{h}<\sum_{g=1}^{q-1}j_{g}\sum_{h=g+1}^{q}j_{h} $
\begin{align}\label{34}
\sum_{i\neq j}e_{i}e_{j}&\leq 2\left[\frac{n(n-1)}{2}e_{avg}^{2} +2\sum_{g=1}^{q-1}j_{g}\sum_{h=g+1}^{q}j_{h}-\left(\sum_{i=1}^{q}j_{i}\right)^{2}\right]\notag\\
&=2\left[\frac{n(n-1)}{2}e_{avg}^{2} -\sum_{i=1}^{q}j_{i}^{2}\right]
\end{align}
Combining \eqref{33} and \eqref{34} we have, 
\begin{align}\label{34}
\sum_{i\neq j}e_{i}e_{j}&\leq n(n-1)e_{avg}^{2} \notag
\end{align}
\end{proof}

\begin{lem}
If we select $ \{ c_{j}\}_{j=1}^{n}, c_{j}\in \{1,-1\} $ such that,
\begin{align}
\sum_{j=1}^{n}c_{j}e_{j}=S_{min}
\end{align}
then,
\begin{align}\label{c}
\sum_{i\neq j}c_{i}c_{j}e_{i}e_{j} <0
\end{align}
\end{lem}
 We can write the left hand side of the equation \eqref{c} as,
\begin{align}\label{fg}
=&\sum_{a_{i},a_{j}\in A^{*}, a_{i}\neq a_{j}} F_{a_{i}}F_{a_{j}}+ \sum_{a_{i}^{c},a_{j}^{c}\in (A^{*})^{c}, a_{i}^{c}\neq a_{j}^{c}} G_{a_{i}^{c}}G_{a_{j}^{c}}\notag\\
&-2(F_{a_{1}}+...+ F_{a_{l}})(G_{a_{1}^{c}}+...+ G_{a_{n-l}^{c}})
\end{align}
For $ l=1 $ the term $\sum_{a_{i},a_{j}\in A^{*}, a_{i}\neq a_{j}} F_{a_{i}}F_{a_{j}} $ does not exist. Similarly for $ n-l=1 $ the term
$\sum_{a_{i}^{c},a_{j}^{c}\in (A^{*})^{c}, a_{i}^{c}\neq a_{j}^{c}} G_{a_{i}^{c}}G_{a_{j}^{c}} $ does not exist. For $ l\geq 2 $ applying  Lemma~\ref{l6} in equation \eqref{fg} we get, 
\begin{align}\label{49}
&\sum_{i\neq j}c_{i}c_{j}e_{i}e_{j}\notag\\
&\leq l(l-1)F_{avg}^{2}+(n-l)(n-l-1)G_{avg}^{2}-2M(M+S_{min})\notag\\
&=\frac{l-1}{l}(M+S_{min})^{2}+\frac{n-l-1}{n-l}M^{2}-2M(M+S_{min})
\end{align}
$ F_{min} $ is the smallest element in the set $ F $, so we can write,
\begin{align}\label{60}
F_{min}\leq \frac{M+S_{min}-F_{min}}{l-1}
\end{align}
Now, applying Lemma~\ref{l5} in \eqref{60} we get $ S_{min}\leq \frac{M}{l-1} $.  

Using \eqref{49} we have,
\begin{align}\label{50}
\sum_{i\neq j}c_{i}c_{j}e_{i}e_{j}\leq &M^{2}\left[\frac{(l-1)^{2}+1}{(l-1)l}+\frac{n-l-1}{n-l} -2\right]\notag\\
&+2MS_{min}\left(\frac{l-1}{l}-1\right)\notag\\
&<0\notag
\end{align}
because, $ \frac{(l-1)^{2}+1}{(l-1)l}\leq 1$ for $l\geq 2 $ and that completes the proof.

\bibliography{ref}
\bibliographystyle{IEEEtran}

\end{document}